\documentclass[11pt]{article}

\usepackage{makeidx}
\usepackage{euscript}
\usepackage{amsmath,amssymb,amsfonts}
\usepackage{amsthm}
\usepackage{latexsym}
\usepackage{subfigure}
\usepackage{graphicx}
\usepackage{url}
\usepackage{times}
\usepackage[scaled=0.92]{helvet} 
\usepackage{fancybox}
\usepackage[margin=1in]{geometry}
\usepackage[noend,noline,ruled,scleft,nofillcomment]{algorithm2e}
\usepackage{algpseudocode}
\usepackage{tikz}
\usetikzlibrary{shapes,snakes}
\usepackage{color}
\usepackage{caption}

\usepackage{wrapfig}

\ifpdf \pdfpageheight11in \pdfpagewidth8.5in \fi

\newtheorem{theorem}{Theorem}
\newtheorem{lemma}[theorem]{Lemma}

\newtheorem{definition}[theorem]{Definition}

\include{amssymb}

\newenvironment{proofof}[1]{\begin{trivlist} \item {\bf Proof
#1:~~}}
  {\qed\end{trivlist}}
\renewenvironment{proofof}[1]{\par\medskip\noindent{\bf Proof of #1: \ }}{\hfill$\Box$\par\medskip}

\newcommand{\etal}{{\em et al. }}
\newcommand{\ie}{{\em i.e. }}
\newcommand{\eps}{\varepsilon}

\newcommand{\E}{\mathsf{E}}

\newcommand{\COMMENTED}[1]{{}}

\newcommand{\Ex}[1]{\ensuremath{\mathbf{E}[#1]}}

\DeclareMathOperator{\polylog}{poly log}

\newcounter{proccnt}

\author{ 
Graham Cormode\thanks{Department of Computer Science, University of
  Warwick, UK. Supported in part by European Research Council grant
  ERC-2014-CoG 647557 and a Royal Society Wolfson Research Merit Award. {\tt g.cormode@warwick.ac.uk}.}
\and
Hossein Jowhari\thanks{Department of Computer Science, University of
  Warwick, UK. Supported by European Research Council grant
  ERC-2014-CoG 647557. 
  {\tt H.Jowhari@warwick.ac.uk }.}
\and
Morteza Monemizadeh\thanks{Rutgers University, Piscataway, NJ, USA. {\tt mortezam@dimacs.rutgers.edu}.}
\and 
S. Muthukrishnan  \thanks{Rutgers University, Piscataway, NJ, USA. {\tt muthu@cs.rutgers.edu}.}
}

\date{\today}

\title{The Sparse Awakens: Streaming Algorithms for Matching Size Estimation in Sparse Graphs}

\begin{document}

\sloppy
\setlength{\abovecaptionskip}{0.1ex}
 \setlength{\belowcaptionskip}{0.1ex}
 \setlength{\floatsep}{0.1ex}
 \setlength{\textfloatsep}{0.1ex}

 \abovedisplayskip.30ex
   \belowdisplayskip.30ex
   \abovedisplayshortskip.30ex
   \belowdisplayshortskip.30ex
   
\captionsetup{belowskip=12pt,aboveskip=4pt}

\begin{titlepage}
  \maketitle
  \thispagestyle{empty}
\addtocounter{page}{-1}

\begin{abstract}
Estimating the size of the maximum matching is a canonical problem in
graph algorithms, and one that has attracted extensive study over a
range of different computational models. 
We present improved streaming algorithms for approximating the size of maximum matching with sparse (bounded arboricity) graphs. 

\begin{itemize}
\item {\em (Insert-Only Streams)}
We present a one-pass algorithm that takes  $O(c\log^2 n)$ space and approximates the size of the  maximum matching in 
graphs with arboricity $c$ within a factor of $O(c)$. This improves significantly upon the state-of-the-art $\tilde{O}(cn^{2/3})$-space streaming algorithms.  

\item {\em (Dynamic Streams)}
Given a dynamic graph stream (i.e., inserts and deletes) of edges of an underlying $c$-bounded arboricity graph,  
we present an one-pass algorithm that uses  space $\tilde{O}(c^{10/3}n^{2/3})$ 
and returns an $O(c)$-estimator for the size of the maximum matching. 
This algorithm improves the state-of-the-art $\tilde{O}(cn^{4/5})$-space algorithms, where the $\tilde{O}(.)$ notation hides logarithmic in $n$ dependencies.

\end{itemize}

 In contrast to the previous works, our results take more advantage of the streaming access to the input and characterize the matching size 
based on the ordering of the edges in the stream in addition to the degree
distributions and structural properties of the sparse graphs. 



\end{abstract}
\end{titlepage}

\section{Introduction}

Problems related to (maximum) matchings in graph have a long history
in Combinatorics and Computer Science.
They arise in many contexts, from choosing which advertisements to display to
online users, to characterizing properties of chemical compounds.
Stable matchings have a suite of applications, from assigning students
to universities, to arranging organ donations.
These have been addressed in a variety of different computation
models, from the traditional RAM model, to more recent sublinear
(property testing) and external memory (MapReduce) models.
Matching has also been studied for a number of classes of input graph:
including general graphs, bipartite graphs, weighted graphs, and
those with some sparsity structure. 

We focus on the streaming case, where each edge is seen
once only, and we are restricted to space sublinear in the size of the graph (ie., no. of its vertices).
In this case, the objective is to find (approximately) the size of the
matching.
Even here, results for general graphs are either weak or make assumptions
about the input or the stream order. In this work, we seek to improve
the guarantees by restricting to graphs that have some measure of
sparsity -- bounded arboricity, or bounded degree.
This aligns with reality, where most massive graphs have
asymptotically fewer than $O(n^2)$ edges.



Recently, Kapralov,  Khanna, and Sudan~\cite{KKS14} developed a streaming algorithm which computes an estimate of matching size for general graphs 
within a factor of $O(\polylog(n))$ in the \emph{random-order} streaming model using $O(\polylog(n))$ space.
In the random-order model, the input stream is assumed to be chosen uniformly at random from the set of all possible permutations of the edges.
Esfandiari \etal \cite{EsfandiariHLMO15} were the first to study streaming
algorithms for estimating the size of matching in bounded arboricity graphs in the \emph{adversarial-order} streaming model, 
where the algorithm is required to provide a good approximation for any ordering of edges. 
Graph arboricity is a measure to quantify the density of a given graph.
A graph $G(V,E)$ has arboricity $c$ if the set $E$ of its edges can be
partitioned into at most $c$ forests. 
Since a forest on $n$ nodes has at most $n-1$ edges, a graph with
arboricity $c$ can have at most $c(n-1)$ edges.
Indeed, by a result of Nash-Williams~\cite{NW61,NW64} this holds for any subgraph of a $c$-bounded arboricity graph $G$.
Formally, the Nash-Williams Theorem~\cite{NW61,NW64} states that $c = \max_{U \subseteq V}\{|E(U)|/(|U| -1)\}$, where $|U|$ and $|E(U)|$ are the
number of nodes and edges in the subgraph with nodes $U$, respectively.
Several important families of graphs have constant arboricity. 
Examples include planar graphs (that have arboricity $c=3$), bounded genus graphs, bounded treewidth graphs, and more generally, 
graphs that exclude a fixed minor.\footnote{It can be shown that for an $H$-minor-free graph, 
the arboricity number is $O(h\sqrt{h})$ where $h$ is the number of vertices of $H$.~\cite{Kost84}} 

The important observation in \cite{EsfandiariHLMO15} is that the size of matching in bounded arboricity 
graphs can be approximately characterized by the number of high degree vertices (vertices with degree above a 
fixed threshold) and 
the number of so called {\em shallow edges} (edges with both low degree endpoints).   
This characterization allows for estimation of the matching size in sublinear space by 
taking samples from the vertices and edges of the graph. The work of
\cite{EsfandiariHLMO15} implements the characterization in $\tilde{O}(cn^{2/3})$ space 
giving a $O(c)$ approximation of the matching size. Subsequent works \cite{BS15,MV16} consider alternative 
characterizations and improve upon the approximation factor however they do not result in 
major space improvements.

\subsection{Our Contributions}

We present major improvements in the space usage of streaming algorithms for sparse graphs ($c$-bounded Arboricity Graphs). 
Our main result is a polylog space algorithm that beats the $n^{\eps}$ space bound of prior algorithms. More precisely, we show: 


%

\begin{theorem} 
\label{thm:arboricity:logn}
Let $G(V,E)$ be a graph with arboricity bounded by $c$. 
Let $S$ be an (adversarial order) insertion-only stream of the edges of the underlying graph $G$. 
Let $M^*$ be the size of the maximum matching of $G$ (or $S$ interchangeably). 
Then, there is a randomized $1$-pass streaming algorithm 
that outputs a  $(22.5c+6)(1+\eps)$-approximation to  $M^*$ with probability at least $1-\delta$ 
  and takes $O(\frac{c}{\eps^2}\log(\frac{1}{\eps})\log cn\log n)$ space.  
\end{theorem}


For the case of dynamic streams
(i.e, streams of inserts and deletes of edges), we design a different algorithm using $\tilde{O}(c^{10/3}n^{2/3})$ space 
which improves the $\tilde{O}(cn^{4/5})$-space dynamic (insertion/deletion) streaming algorithms of \cite{BS15,CCEHMMV16}. 
The following theorem states this result (proved in
Section~\ref{sec:dynamicalg}).

\begin{theorem} 
\label{thm:insert:delete}  
Let $G(V,E)$ be a graph with the arboricity bounded by $c$. 
Let $M^*$ be the size of the maximum matching of $G$. 
Let $S$ be a dynamic stream of edge insertions and deletions of the underlying graph $G$ of length at most $O(cn)$ . 
Let $\beta=\mu( (2\mu)/(\mu - 2c+1)+1)$ where $\mu > 2c$. 
Then, there exists a streaming algorithm that takes $O(\frac{\beta^{4/3}(nc)^{2/3}}{\eps^{4/3}})$ space in expectation 
and outputs a $(1+ \eps)\beta$ approximation of $M^*$ with probability at least $0.86$. 
\end{theorem}


Our algorithms for bounded arboricity graphs are based on two novel {\em streaming-friendly} characterizations of the maximum maching size. 
The first characterization is a modification of the characterization in \cite{EHLMO15} which 
approximates the size of the maximum matching by $h_\mu+s_\mu$ where $h_\mu$ 
is defined as the number of high degree vertices (vertices with degree more than a threshold $\mu$) and $s_\mu$ is the number of shallow edges 
(edges with low degree endpoints). While $h_\mu$ can be easily approximated 
by sampling the vertices and checking if they are high degree or not,
approximating $s_\mu$ in sublinear space is a challenge because in one
pass we cannot determine 
if a sampled edge is shallow or not. The work of \cite{EHLMO15}
resolves this issue by sampling the edges at a high rate and manages to
implement their characterization in $\tilde{O}(cn^{2/3})$ space for adversarial 
insert-only streams. 

To bring the space usage down to $\tilde{O}(c^{2.5}n^{1/2})$ (for insert-only streams), we modify 
the formulation of the above characterization. We still need to approximate $h_\mu$ but instead of $s_\mu$ we approximate $n_L$ the number of non-isolated 
vertices in the induced subgraph $G_L$ defined over the low degree
vertices. Note that $s_\mu$ is the number of edges in $G_L$.
This subtle change of definition turns out to be immensely helpful.
Similar to $h_\mu$ we only need to sample the nodes and check if their degrees
are below a certain threshold or not, however we carry the
additional constraint that we have to avoid counting the nodes in $G_L$ that are isolated (have only high degree nodes as neighbors). 
To satisfy this additional constraint, our algorithm stores the
neighbors of the sampled vertices along with a counter for each that
maintains their degree in the rest of the stream.
Although we only obtain a lower bound on the degree of the neighbors,
as it turns out the lower bound information on the degree is still
useful because we can ensure the number of false positives
that contribute to our estimate is within a certain limit.
As result, we can approximate 
$h_\mu+n_L$ using $\tilde{O}(cn^{1/2})$ space which gives a $(2c+1)(2c+2)$ approximation
of the maximum matching size after choosing appropriate values for $\mu$ and other parameters.
This characterization is of particular importances, as it can be adapted to work under edge deletions as well as long as the number of deletions is bounded by $O(cn)$.    
Details
of the characterization and the associated algorithms are given in Lemma \ref{lem:hc+1} and Section \ref{sec:sqrtn.alg}.

To obtain a $\polylog(n)$ space algorithm (and prove the claim of Theorem \ref{thm:arboricity:logn}), 
we give a totally new characterization. This characterization, unlike the previous ones that only depend on the parameters 
of the graph, also takes  the ordering of the edges in the stream into account. Roughly speaking, we characterize the size of a maximum matching by the number of 
edges in the stream that have few neighbor edges in the rest of the
stream. To understand the connection with maximum matching, consider
the following simplistic special case.
Suppose the input graph $G$ is a forest composed of $k$ disjoint
stars.
Observe that the maximum matching on this graph is just to pick one
edge from each star. 
We relate this to a combinatorial characterization that arises from
the sequence of edges in the stream: 
no matter how we
order the edges of $G$ in the stream, from each star there is exactly one edge that has no neighboring edges 
in the remainder of the stream (in other words, the last edge of the star in the stream). 
Our characterization generalizes this idea to graphs with arboricity bounded by $c$ by counting the {\it $\alpha$-good edges}, 
{\ie}edges that have at most $\alpha=6c$ neighbors in the remainder of
the stream.
We prove this characterization gives an $O(c)$ approximation of the
maximum matching size. More important, a nice feature of this characterization is that it can be implemented in $\polylog(n)$ space if one allows a $1+\eps$ approximation. 
The implementation adapts an idea from the well-known $L_0$ sampling algorithm.       
It runs $O(\log n)$ parallel threads each sampling the stream at a different rate.
At the end, a thread ``wins'' that has sampled
roughly $\Theta(\log n)$ elements from the $\alpha$-good edges (samples the edges with a rate of $\frac{\log n}{k}$ where $k$ is the number of $\alpha$-good edges). 
The threads that under-sample will end up with few edges or nothing while the ones that have oversampled will keep too many $\alpha$-good edges and will be terminated as result.

Table~\ref{fig:results} summarizes the known and new results for estimating the size of a maximum matching.

\begin{table*}[t]
\begin{center}
\renewcommand{\arraystretch}{1.3}
\begin{tabular}{|l||l|l|l|l|l}
\hline 
Reference  & Graph class & Stream  & Approx. Factor $^*$ & Space Bound$^{**}$   \\
\hline
\cite{KKS14}     & General                         & Random    Order    & $O(\polylog(n))$              &  $O(\polylog(n))$  \\

\cite{EHLMO15} & Arboricity $\le c$ &Insert-Only & $5c+9$ & $\tilde{O}(cn^{2/3})$ \\
\cite{MV16} & Arboricity $\le c$ &Insert-Only & $c+2$ & $\tilde{O}(cn^{2/3})$ \\
\cite{BS15,CCEHMMV16} & Arboricity $\le c$ &Insert/Delete & $O(c)$ & $\tilde{O}(cn^{4/5})$  \\
\hline
This paper  & Arboricity $\le c$ & Insert-Only & $(2c+1)(2c+2)$ &$\tilde{O}(c^{2.5}\sqrt{n})$  \\
This paper  & Arboricity $\le c$ & Insert/Delete & $(2c+1)(2c+2)$ &$\tilde{O}(c^{10/3}n^{2/3})$   \\
This paper  & Arboricity $\le c$ & Insert-Only & $22.5c+6$ &$\tilde{O}(c\log^2 n)$ \\
\hline
\end{tabular}
\end{center}
\caption{Known results for estimating the size of a maximum matching in data streams. 
(*) In some entries, a $(1+\eps)$ multiplicative factor has been
  suppressed for concision. (**) The $\tilde{O}(.)$ notation hides logarithmic in $n$ dependencies. }
\label{fig:results}
\end{table*}

\subsection{Further Related Streaming Work}
  In the classical offline model, where we assume we have enough space to store all vertices and 
edges of a graph $G=(V,E)$, the problem of computing the maximum matching of $G$ 
has been extensively studied.
The best result in this model is the $30$-years-old algorithm due to Micali and Vazirani~\cite{MV80}
with running time $O(m \sqrt{n})$, where $n=|V|$ and $m=|E|$.
A matching of size within $(1-\eps)$ factor of a maximum cardinality matching can be found in $O(m/\eps)$ time \cite{HopcroftK73,MV80}. 
Very recently, Duan and Pettie \cite{DP14} develop a $(1-\eps)$-approximate maximum weighted matching algorithm in time  $O(m/\eps)$.


The question of approximating the maximum cardinality matching has
been extensively studied in the streaming model. An $O(n)$-space greedy algorithm trivially obtains a maximal matching, which is a $2$-approximation for the maximum cardinality matching~\cite{FKMSZ05}. 
A natural question is whether one can beat the approximation factor of the greedy algorithm with $O(n\polylog(n))$ space. Recently, 
it was shown that obtaining an approximation factor better than $\frac{e}{e-1}\simeq 1.58$ in one pass requires $n^{1+\Omega(1/\log\log n)}$ space \cite{GoelKK12,Kap13},
even in bipartite graphs and in the \emph{vertex-arrival} model, where the vertices arrive in the stream together with their incident edges. 
This setting has also been studied in the context of \emph{online algorithms}, where each arriving vertex has to be either matched or discarded irrevocably upon arrival. 
Seminal work due to Karp, Vazirani and Vazirani \cite{KVV90} gives an online algorithm with $\frac{e}{e-1}$ approximation factor in this model. 

Closing the gap between the upper bound of $2$ and the lower bound of $\frac{e}{e-1}$ remains one of the most appealing open problems in the graph streaming area (see~\cite{P60}). 
The factor of 2 can be improved on if one either considers the random-order model or allows for two passes \cite{KMM12}. 
By allowing even more passes, the approximation factor can be improved to multiplicative $(1-\epsilon)$-approximation via finding and applying augmenting paths with
successive passes~\cite{McGregor:05,McGregor:09,EKMS12,EggertKMS12,AGM12a}.

Another line of research~\cite{FKMSZ05,McGregor:05,Zelke12,EpsteinLMS11} has explored the question of approximating 
the maximum-weight matching in one pass and $O(n\polylog(n))$ space. Currently, the best known approximation factor equals $4 + \eps$ (for any positive constant $\eps$)~\cite{CS14}.

\section{Preliminaries and Notations}
%
%
%

%
%
Let $G(V,E)$ be an undirected unweighted graph with $n=|V|$ vertices and $m=|E|$ edges. 
For a vertex $v\in V$, let $\deg_G(v)$ denote the degree of vertex $v$ in $G$. 
%
%
%
A \emph{matching} $M$ of $G$ is a set of pairwise non-adjacent edges, i.e., no two edges share a common edge. 
Edges in $M$ are called \textit{matched} edges; the other edges are called \textit{unmatched}.
A \emph{maximum matching} of graph $G(V,E)$ is a matching of maximum size. 
Throughout the paper, when we fix a maximum matching of $G(V,E)$, we denote it by $M^*$.
A matching $M$ of $G$ is \emph{maximal} if it is not a proper subset of any other matching in graph $G$.
Abusing the notation, we sometimes use $M^*$  and $M$ for the size the maximum and maximal matching, respectively. 
It is well-known (see for example \cite{LP86}) that the size of a maximal matching  is at least half of the size of a maximum matching, i.e., 
$M \ge M^*/2$. Thus, we say a maximal matching is a $2$-approximation of a maximum matching of $G$. 
It is known \cite{LP86} that the simple greedy algorithm, where we include the newly arrived edge if none of its endpoint are already matched, returns a maximal matching. 



\section{Algorithms for Bounded Arboricity Graphs}
Throughout this section, $h_\mu$ denotes the number of
vertices in graph $G=(V,E)$ that have degree above $\mu$. Let 
$G_L=(V',E')$ be an induced subgraph of $G$ where 
$V'=\{v | \deg_G(v) \le \mu \}$ and $(u,v) \in E'$ iff $u$ and $v$ are
both in $V'$. Note that $G_L$ might have isolated vertices. In 
the following we let $M_\mu$ denote the size of maximum matching
in $G_L$. 

\subsection{Characterization lemmas}

\begin{lemma} [\cite{EHLMO15}]
\label{lem:EHL+15}
For a $c$-bounded arboricity graph $G(V,E)$ and  $\mu > 2c $, we have 
$ h_\mu \le \frac{2\mu}{\mu - 2c+1} M^*$. 
\end{lemma}

\begin{lemma}
\label{lem:hc+1}
For a $c$-bounded arboricity graph $G(V,E)$ and  $\mu > 2c $, we have 
$$ M^* \le h_{\mu}+M_{\mu} \le \left( \frac{2\mu}{\mu - 2c+1}+1\right)M^* \enspace .$$
\end{lemma}

\begin{proof} The lower bound is easy to see: every edge of a maximum matching either has an endpoint with
degree more than $\mu$ or both of its endpoints are vertices with
degree at most $\mu$.
The number of matched edges of the first type are bounded by $h_\mu$ whereas
the number of matched edges of the second type are bounded by
$M_\mu$. 

To prove the upper bound, we use the fact  $M_{\mu} \le M^*$ and Lemma \ref{lem:EHL+15}.  
\end{proof}
\begin{definition}\label{alphagood}
Let $S=(e_1,\ldots,e_m)$ be a sequence of edges. 
We say the edge $e_i=(u,v)$ is {\em $\alpha$-good} with respect
to $S$ if $\, \max \{d_i(u),d_i(v)\} \le \alpha$ where $d_i(x)$ 
is defined as the number of the neighbors of $x$ that 
appear after the $i$-th location in the stream. 
\end{definition}

\begin{lemma} 
\label{lem:E_alpha}
Let $\mu > 2c$ be a (large enough) parameter. 
Let $E_{\alpha}$ be the set of $\alpha$-good edges in an edge stream
for a graph with arboricity at most $c$.
We have:

 $$\left(\frac12-\frac{c}{\mu+1}\right) M^* \; \le \; |E_{\alpha}| \; \le \; \left(\frac{5}{4}\alpha+2\right) M^* , $$
 
 where $\alpha =\max\{\mu-1,\frac{4c(\mu+1)}{\mu+1-2c}\}$. In particular for $\mu=6c-1$, we have 
 $$M^* \; \le \; 3|E_{6c}| \; \le \;  (22.5c+6)M^* $$
\end{lemma}

\begin{proof}
First we prove the lower bound on $|E_\alpha|$.
In particular we show a relation involving
the number of edges where both endpoints have low degree,
$s_\mu = |\{ e = (u,v) | e \in E, \deg(u) \leq \mu, \deg(v) \leq \mu\}|$:

$$\left(\frac12-\frac{c}{\mu+1}\right) h_{\mu} + s_\mu \; \le \; |E_{\alpha}|.$$

The claim in the lemma follows from the relatively loose bound that $M^*\le h_\mu+s_\mu$.
 Let $H$ be the set of vertices in the graph with
degree above $\mu$ and let $L=V \setminus H$. Recall that 
$h_\mu=|H|$.  
Let $H_0$ be 
the vertices in $H$ that have no neighbor in $L$, and let $H_1 = H \setminus H_0$. 
First we notice that $|H_1| \ge (1-\frac{2c}{\mu})|H|$. To see this, let $E'$ be the 
edges with at least one endpoint in $H_0$.
By definition, every node in $H_0$ has degree at least $\mu+1$,
so we have $ |E'| \ge \frac{\mu+1}{2}|H_0|$. 
At the same time, the total number of edges in the subgraph induced by
the nodes $H$ is at most
$c(|H|-1)$, using the arboricity assumption. Therefore, $$c(|H|-1) \ge |E'| \ge \frac{\mu+1}{2}|H_0| $$   
It follows that $|H_0| \le \frac{2c}{\mu+1}(|H|-1)$ which further implies
that
\begin{equation}
  |H_1| \ge (1-\frac{2c}{\mu+1})|H| = (1-\frac{2c}{\mu+1})h_\mu.
\label{eq:sizeofh1}
\end{equation}

Now let $d_{H}(v)$ be the degree of $v$ in the induced subgraph $H$. We have 
$\sum_{v \in H_1} d_H(v) \le 2c|H|$, again using the arboricity
bound and the fact that summing over degrees counts each edge at most
twice. Therefore, taking the average over nodes in $H_1$,
$$\overline{d_H}(v) \le \frac{2c}{1-\frac{2c}{\mu+1}}$$ for $v \in H_1$.
Consequently, at least half of the vertices in $H_1$ have their $d_H$ bounded
by $\frac{4c(\mu+1)}{\mu+1-2c}$ (via the Markov inequality).
Let $H'_1$ be those vertices. 
For each $v \in H'_1$ we find an $\alpha$-good edge. Let $e^*=(v,u)$ be
the last edge in the stream where  $u\in L$.
Then, there cannot be too many edges that neighbor $(v,u)$ and come after
it in the stream:
the total number of edges that share an endpoint with
$e^*$ in the rest of the stream is 
bounded by $\max\{\mu-1,\frac{4c(\mu+1)}{\mu+1-2c}\}$.
Consequently, for 
$\alpha = \max\{\mu-1,\frac{4c(\mu+1)}{\mu+1-2c}\}$, we have $|E_\alpha| \ge
(\frac12-\frac{c}{\mu+1})h_{\mu}$, based on the set of $|H_1|/2$ edges in
$H_1'$ and using \eqref{eq:sizeofh1}.
For $\alpha \ge \mu$, $E_\alpha$ also contains the disjoint set of edges
from $L \times L$, which are all guaranteed to be $\alpha$-good since
both their endpoints have degree bounded by $\mu$.
Therefore $|E_\alpha| \ge s_\mu+(\frac12-\frac{c}{\mu+1})h_{\mu}$ as claimed. 

To prove the upper bound on $|E_\alpha|$, we notice that the subgraph containing only the edges in $E_\alpha$ has degree at most $\alpha+1$. Such a graph has a matching size of at least $\frac{4|E_\alpha|}{5(\alpha+1)+3}$ \cite{Han08}. It follows that $|E_\alpha| \le \frac{5\alpha+8}4M^*$. This
finishes the proof of the lemma.  
\end{proof}

Last, we note that in the special case of trees (or more generally,
graph streams which represent forests), a tighter approximation bound
follows (for which our algorithms specified below will also apply). 

\begin{lemma} For trees we have $ M^* \le |E_1| \le 2M^*$.
\end{lemma}
\begin{proof} 
Let $T=(V,E)$ be a tree with maximum matching size $M^*$.
The upper bound follows by considering $E_1$: the subgraph of $T$
formed by $E_1$ has degree  at most 2, and since we are considering
trees, the set $E_1$ can  have no cycles and so consists of paths. 
Hence, $|E_1| \le 2M^*$. 

For the lower bound, we use induction on the number of nodes.
Suppose the claim $|E_1| \ge M^* $ is true
for all trees on $n$ nodes. We want to show that the claim remains true for trees
on $n+1$ nodes. The base case $n=2$ is trivially true. 
Given a tree $T=(V,E)$ with $n+1$ nodes,
there is always a leaf $w \in V$, that is connected to a node $u$ where $u$ has at most one non-leaf neighbor. If we remove the edge $(u,w)$ from the tree, we get a tree $T'$
with $n$ nodes and by our induction hypothesis, $|E_1(T')| \ge M^*(T')$ no matter how the
stream is ordered.
Fix some ordering of the stream for $T'$. We claim after inserting
the edge $(u,w)$ in the stream (anywhere) we will have $|E_1(T)| \ge M^*(T)$.
Why? We have two cases to consider. 
\begin{enumerate}
\item $w$ has no sibling in $T$.
In this case $E_1(T)=E_1(T')\cup \{(u,w)\}$. This is because $u$ must have
been a leaf in $T'$ and as result adding $(u,w)$ does not cause any
other edge to lose the $1$-goodness property. If follows that the size
of $E_1$ increases by 1 while $M^*$ increases by at most 1.

\item $w$ has a sibling. In this case for sure $M^*$ does not increase.
Although there may be a concern that the size of $E_1$ could drop,
below we show that adding a leaf to the stream of the edges
of a tree does not cause the size of the 
set $E_1$ to drop. This is enough to show that in this case as well 
$|E_1(T)| \ge M^*(T)$.  
\end{enumerate}

To see why adding a leaf to the stream of edges cannot
reduce $|E_1|$, assume we insert an edge $e=(u,w)$
in the stream where $w$ is a newly added leaf. If $u$ has no $1$-good edges 
 incident on it, then $E_1$ remains as it was.
If $u$ has one $1$-good edge on it and adding $e=(u,w)$ causes 
it to be kicked out of $E_1$,
it means $e$ is admitted as a new member of $E_1$. So the loss is accounted for. 
If $u$ has two $1$-good edges on it, say $e_1$ and $e_2$,
we show adding $e$ cannot cause them to be ejected from  $E_1$.
To see this, assume to the contrary that it could, 
and suppose (without loss of generality) the edges come in the following order in the stream: 
$\ldots, e_2, \ldots, e_1,\ldots,e,\ldots$.
The edge $e_2$ clearly cannot be part of $E_1$.
But $e_1$ must have a neighboring edge $e_3$ incident on $u$,
that follow it. But that means $e_2$ was already out 
before adding $e$ to the stream. A contradiction.
Finally we note that $u$ cannot have more than two $1$-good 
edges on it. This finishes the proof.
\end{proof}


\subsection{$\tilde{O}(\sqrt{n})$ space algorithm for insert-only streams}
\label{sec:sqrtn.alg}
In this section, first we present Algorithm~\ref{alg:Mc} that estimates $M_{\mu}+h_{\mu}$ and prove the following theorem.

\begin{theorem} 
\label{thm:matching:root:n}  
Let $G(V,E)$ be a graph with the arboricity bounded by $c$. 
Let $S$ be an (adversarial order) insertion-only stream of the edges of the underlying graph $G$. 
Let $\beta=\mu( (2\mu)/(\mu - 2c+1)+1)$ where $\mu > 2c$. 
Then, there exists an insertion-only streaming algorithm (Algorithm \ref{alg:Mc}) that takes $O(\frac{\beta\sqrt{cn}}{\eps}\log n)$ space in expectation 
and outputs a $(1+ \eps)\beta$ approximation of $M^*$ with probability at least $0.86$, 
where $M^*$ is a maximum matching of $G$. 
\end{theorem}

\begin{algorithm}[th]
  {\bf Initialization:} {Each node is sampled to set $S$ with probability $p$ (determined below). }

\medskip
{\bf Stream Processing:}

\ForAll{edges $e = (u,v)$ in the stream}{
\If{$u \in S$ or $v \in S$}{
    store $e$ in $H$
    \lIf{$u \in S$}{ increment $d(u)$ {\bf else} increment $l(u)$}
    \lIf{num$v \in S$}{ increment $d(v)$ {\bf else} increment $l(v)$}
}
}

\medskip
{\bf Post Processing:}

Let $S_1 = \{v \in S | d(v) \leq \mu, 
\exists w \in \Gamma(v): d(w) + l(w) \leq \mu\}$

Let $S_2$
be the set of vertices $\{v | v \in S, d(v) > \mu\}$

\Return{$s=(|S_1|+|S_2|)/p$}

 \caption{Estimate-$M_{\mu}+h_{\mu}$}
 \label{alg:Mc}
\end{algorithm}

For each $w$ in $\Gamma(S)$ (the set of neighbors of nodes in $S$),
the algorithm maintains $l(w)$, the number of occurrences of $w$
observed since (the first) $v \in S$ such that $w \in \Gamma(\{v\})$ was added.
Note that in this algorithm, $l(w)$ is a lower bound on the degree of $w$. 
For the output, $S_1$ is the subset of nodes in $S$ whose degree is
bounded by $\mu$ and 
  additionally there is at least one neighbor of $v$, $w$, whose
  observed degree ($d(w)$ or $l(w)$) is at most $\mu$. 
Meanwhile, $S_2$ is the ``high degree'' nodes in $S$. 

\begin{lemma}
\label{lem:Mc} Let $\eps \in (0,1)$ and $\beta=\mu( \frac{2\mu}{\mu - 2c+1}+1)$. 
Algorithm 1 outputs $s$ where 
$(1-\eps)M^* \; \le s \le \; (1+\eps)\beta M^*$ 
 with probability at least 
$1-e^{\frac{-\eps^2M^*p}{4\beta^2}}$.
\end{lemma}

\begin{proof}
 First we prove the following bounds on $\E(s)$. 
 $M_{\mu}+h_{\mu} \le \E(s)\le \mu(M_{\mu}+h_{\mu})$. 
Let $L$ be the set of vertices in $G$ that have degree at most $\mu$ and let $G_{L}$ be the induced graph on $L$. Let $H=V \setminus L$. Note that $G_{L}$ might have isolated
vertices. Let $N$ be the non-isolated vertices in $G_{L}$. It is clear that if the algorithm samples $v \in N$, 
$v$ will be in $S_1$. Likewise, if it samples a vertex $w \in H$, $w$
will be in $S_2$.
Given the fact that $|H| = h_\mu$ and $|N| \ge M_{\mu}$, this proves the lower bound on $\E(s)$. 

The expectation may be above $M_u$, as 
the algorithm may pick an isolated
vertex in $G_{L}$ (a vertex that is {\em only} connected to the high-degree vertices) and include it in $S_1$ because 
one of its high-degree neighbours $w$ was identified as low degree,
i.e., $w \in \Gamma(S)$ and $l(w) \leq \mu$ but $w \in H$. 
Let $u \in H$  and let $U=\{a_1,\ldots,a_\mu\}$ be the last $\mu$
neighbours of $u$ according to the ordering of the edges in the
stream. The algorithm can only identify $u$ as low
degree when it picks a sample from $U$ and no samples from 
$\Gamma(u)\setminus U$. This restricts the number of {\em unwanted} isolated vertices to at most $\mu h_{\mu}$. Together with the fact
that $|N| \le \mu M_{\mu}$, it establishes the upper bound on $\E(s)$.

 Now using a Chernoff bound, 
  $\Pr[|s -\E(s)| \le \lambda \E(s)] \le e^{\frac{-\lambda^2(M_{\mu}+h_{\mu})p}{4}} \le e^{\frac{-\lambda^2M^*p}{4}}  $.
  Therefore with probability at least $1-e^{\frac{-\lambda^2M^*p}{4}}$,
  \begin{equation}
 (M_{\mu}+h_{\mu}) -\lambda \mu(h_{\mu}+M_{\mu}) \le s \le \mu(1+\lambda)(M_{\mu}+h_{\mu})
  \end{equation}
  
Setting $\lambda= \frac{\eps}{\beta}$ and putting this and Lemma \ref{lem:hc+1} together, we derive the statement of the lemma.
\end{proof}
\begin{algorithm}[H]
  {
{\bf Initialization:}
    Let $\eps \in (0,1)$ and $t= \lceil \frac{\beta\sqrt{8nc}}{\eps}
  \rceil$ where $\beta$ is as defined in Lemma \ref{lem:Mc}.

{\bf Stream Processing:}  
  Do the following tasks in parallel:
  \begin{enumerate}
  \item Greedily keep a maximal matching of size at most $r\le t$
    (and terminate this task if this size bound is exceeded).
  \item Run the Estimate-($M_{\mu}+h_{\mu}$) procedure
    (Algorithm~\ref{alg:Mc}) with parameter $p \ge \frac{8}{\lambda^2 t}$ where $\lambda = \frac{\eps}{\beta}$.
 
  \end{enumerate}

  {\bf Post processing}: If $r < t$ then output $2r$ as the estimate for $M^*$, otherwise output the result of the Estimate-($M_{\mu}+h_{\mu}$) procedure.   
 
  }
 \caption{Estimate-$M^*$}
 \label{alg:Mc2}
\end{algorithm}

\begin{proofof}{Theorem \ref{thm:matching:root:n}} 
Suppose $M^* < t$. Clearly the 
size of the maximal matching $r$ obtained by the first task will
be less than $t$. In this case, $M^* \le M' \le 2M^*$.
Now suppose $M^* \ge t.$ By Lemma \ref{lem:hc+1},  we will have $M_{\mu}+h_{\mu} \ge t$ and hence by Lemma
\ref{lem:Mc}, with probability at least $1-e^{-2} \ge 0.86$, the output of the algorithm will be within the promised bounds. The expected space of the algorithm is $O((t+pnc)\log n)$. 
Setting $t=\beta\sqrt{8nc}/\eps$ to balance the space costs, the space complexity of the
algorithm will be $O(\frac{\beta\sqrt{cn}}{\eps}\log n)$ as claimed. 
\end{proofof}


\subsection{$O(n^{2/3})$ space algorithm for insertion/deletion streams} 
\label{sec:dynamicalg}
Algorithms~\ref{alg:Mc} and \ref{alg:Mc2} form the basis of our
solution in the more general case where the stream contains
deletions of edges as well.
In the case of Algorithm~\ref{alg:Mc}, the algorithm
has to  maintain the induced subgraph on $S$ and the edges of the cut
$(S,\Gamma(S))$. However if we allow arbitrary number
of insertions and deletions, the size of the cut $(S,\Gamma(S))$ can grow as large as $O(n)$ even when $|S|=1$. This is because 
each node at some intermediate point could become high degree and then
lose its neighbours because of the subsequent deletion of edges.
Therefore here in order to limit the space
usage of the algorithm, we make the assumptions
that number of deletions is bounded by $O(cn)$. Since
the processed graph has arboricity at most 
$c$ this forces the number of insertions
 to be $O(cn)$ as well. Under this 
assumption, if we pick a random vertex, still, in expectation 
the number of neighbours is bounded by $O(c)$.   

Another complication arises from the fact that, with 
edge deletions, a vertex added to $\Gamma(S)$ might become isolated
at some point. In this case, we discard it from $\Gamma(S)$. Additionally
for each vertex in $S \cup \Gamma(S)$, the counters $d(v)$ (or $l(v)$
depending on if it belongs to $S$ or $\Gamma(S)$) can be
maintained as before. The space complexity of the algorithm remains
$O(pnc\log n)$ in expectation as long as the arboricity factor remains
within $O(c)$ in the intermediate graphs. In the case of
Algorithm~\ref{alg:Mc2}, we need to keep a maximal matching of size $O(t)$. This
can be done in $O(t^2)$ space using a randomized algorithm
\cite{CCEHMMV16}. Setting $t$ at $(\frac{8\beta
  nc}{\eps^2})^{1/3}$ to rebalance the space costs, 
  we obtain the result of Theorem \ref{thm:insert:delete}.


\subsection{The $O(\log^2n)$ space algorithm for insert-only streams}
\label{sec:logalg}
In this section we present our polylog space algorithm by
presenting an algorithm for estimating $|E_\alpha|$ within $(1+\eps)$
factor. 
Our algorithm is similar in spirit to the known $L_0$ sampling strategy. It runs $O(\log n)$ parallel threads each sampling the
stream at a different rate.
At the end, a thread ``wins'' that has sampled
roughly $\Theta(\log n)$ elements from $|E_\alpha|$ (samples the edges with a rate of $\frac{\log n}{|E_\alpha|}$). The threads that
 under-sample will end up with few edges or nothing while the ones that have oversampled will keep too many elements of $E_\alpha$ and will be aborted as result.
 
 First we give a simple procedure (Algorithm~\ref{alg:ALPHA}) that is the building block of the algorithm.

\begin{algorithm}[bh]
  { 
 {\bf Initialization}: given the edge $e=(u,v)$ in the stream, let $r(u)=0$ and $r(v)=0$.
  
\ForAll{subsequent edges $e' = (t,w)$}{
  \lIf{$t=u$ \textrm{or} $w=u$}{increment $r(u)$}
  \lIf{$t=v$ \textrm{or} $w=v$}{increment $r(v)$}
  \lIf{$\max \{r(u),r(v)\} > \alpha$}{terminate and report NOT
    $\alpha$-good}
}
  }
 \caption{ The $\alpha$-good test }
 \label{alg:ALPHA}
\end{algorithm}

Now we present the main algorithm (Algorithm~\ref{alg:ealpha}) followed by its analysis.  

\begin{algorithm}[th]
  {
    {\bf Initialization: } $\forall i. X_i = \emptyset$ \Comment{$X_i$ represents the current set
  of sampled  $\alpha$-good edges.}

    \medskip
{\bf Stream Processing:}

\ForAll{levels $i \in \{0,1,\ldots,[\lfloor \log_{1+\eps}m \rfloor \}$ in parallel}{

\ForAll{edges $e$}{
%

  Feed $e$ to the active $\alpha$-good tests and update $X_i$
  
  With probability $p_i=\frac{1}{(1+\eps)^i}$ 
 add $e$ to $X_i$ and start a $\alpha$-good test for $e$. 
 
  Let $|X_i|$ be 
  the number of active $\alpha$-good tests within this level.

  \lIf{$|X_i|
 > \tau= \frac{64\alpha^2\log n}{c\eps^2}$}{ terminate level $i$}
 }
}
      
  \medskip
  {\bf Post processing:}

  \eIf{$|X_0| \le \tau$}{\Return{$|X_0|$}}(\Comment{$|X_0| > \tau$\hspace{5in}}){
    let $j$ be the smallest integer such that $X_{j}\le \frac{8\log  n(1+\eps)}{\eps^2}$ and the $j$-th level was not terminated\;
    \lIf{there is no such
      $j$}{\Return{FAIL} {\bf else} \Return{$\frac{X_{j}}{p_{j}}$}}
 }
}
 \caption{An algorithm for approximating $|E_\alpha|$}
 \label{alg:ealpha}
\end{algorithm}

\begin{lemma} 
\label{lem:E_alpha_alg} 
With high probability, 
Algorithm \ref{alg:ealpha} outputs a $1\pm O(\eps)$ approximation of $|E_\alpha|$.  
 \end{lemma}
 
 \begin{proof} 
 It is clear that if $|X_0| \le \tau$ then $X_0=E_\alpha$ and
 the algorithm makes no error. In case $|X_0| > \tau$, we
 claim that $|E_\alpha| > \frac{c}{2\alpha^2} \tau$. To prove this let
 $t$ be the time step where $|X_0|$ exceeds $\tau$ and 
 let $G_t=(V,E^{(t)})$ be the graph where 
 $E^{(t)}=\{e_1,\ldots,e_t\}$. Clearly $M^*(G) \ge M^*(G_t)$ because the size of matching only increases.
 Abusing the notation, let $E_\alpha(G_t)$ denote the set of 
 $\alpha$-good edges at time $t$.   
 By Lemma \ref{lem:E_alpha}, we have $$\tau < |E_\alpha(G_t)| \le \left(\frac54 \alpha +2\right) M^*(G_t) \le 4\alpha M^*(G) \le \frac {2\mu}{\mu-2c}4\alpha|E_{\alpha}| \le \frac{2\alpha^2}{c} |E_{\alpha}|.$$ This proves the claim.
 
 Therefore in the following we assume that 
 $|E_\alpha| > \frac{c}{2\alpha^2} \tau $. Let $\tau'=\frac{8\log n}{\eps^2}$ and let $i^*$ be the integer such that $(1+\eps)^{i^*-1}\tau' \le |E_\alpha| \le (1+\eps)^{i^*}\tau'$. 
 
 Assuming the ${i^*}$-th level does not terminate before the end, we have $ \frac{\tau'}{(1+\eps)} \le \E[|X_{i^*}|] \le \tau'$.
 By a Chernoff bound, for each $i$ we have (again assuming we do not
 terminate the corresponding level)
 $$\Pr[||X_{i}|-\E(|X_i|)| \ge \eps \E(|X_i|)] \le \exp{\left(-\frac{\eps^2p_i|E_{\alpha}|}{4}\right)}.$$ 
 Therefore, 
 $
 \Pr[||X_{i^*}|-\E(|X_{i^*}|)| \ge \eps \E(|X_{i^*}|)] \le \exp{ ( -\frac{\eps^2|E_\alpha|}{2(1+\eps)^{i^*}} ) } \le \exp{(\frac{2\log n}{1+\eps})} \le O(n^{-1})
 $.  
 
As a result, with high probability $|X_{i^*}| \le \frac{8\log n(1+\eps)}{\eps^2}$. Moreover for all $i < i^{*}-1$, the corresponding levels either terminate prematurely or in the end we will have $|X_{i}| > \frac{8\log n(1+\eps)}{\eps^2}$ with high probability. Consequently
 $j \in \{i^*,i^*-1\}$. It remains to prove that runs
 corresponding to $i^*$ and $i^*-1$ will survive until the end with high probability. We prove this for $i^*$. The case of $i^*-1$ is similar. 

 Consider a fixed time $t$ in the stream and let $X_{i^*}^{(t)}$ be the set of sampled $\alpha$-good edges
 at time $t$ corresponding to the $i^*$-th level. Note 
 that $X_{i^*}^{(t)}$ contains the a subset of
 $\alpha$-good edges with respect to the stream $S_t =(e_1,\ldots,e_t)$. From the 
 definition of $i^*$ and our earlier observations we have 
$$\E[|X_{i^*}^{(t)}|] = \frac{|E_\alpha(G_t)|}{(1+\eps)^{i^*}} \le \frac{2\alpha^2|E_\alpha|}{c(1+\eps)^{i^*}} \le \frac{2\alpha^2\tau'}{c} $$  
By the Chernoff inequality for $\delta \ge 1$, 
 $$\Pr\left[|X_{i^*}^{(t)}| \ge (1+\delta)E(|X_{i^*}^{(t)}|)=\tau\right]\le \exp{\left( \frac{-\delta}{3} E(|X_{i^*}^{(t)}|)\right)}.$$ 
 From $\delta=\frac{\tau }{E(|X_{i^*}^{(t)}|)}-1 =\frac{\tau (1+\eps)^{i^*}}{|E_\alpha(G_t)|}-1$, we get 
 
 $$\Pr\left[|X_{i^*}^{(t)}| \ge \tau\right]\le \exp{\left( \frac{-\tau}{3} + \frac{|E_\alpha(G_t)|}{(1+\eps)^{i^*}}\right)} \le \exp{\left( \frac{-\tau}{3} +\frac{2\alpha^2\tau'}{c} \right)}$$

For $\tau \ge \frac{8\alpha^2\tau'}{c}$, the term inside the
exponent is smaller than $-2\log n$. It also satisfies $\delta \ge 1$. After applying the union bound, for all $t$ the size of $X_{i^*}^{(t)}$ is bounded by $\tau = \frac{64\alpha^2\log n}{c\eps^2}$ with high probability.  This finished the proof of
the lemma. 
 \end{proof}

Next, putting everything together, we prove Theorem \ref{thm:arboricity:logn}. 
  
\begin{proofof}{Theorem \ref{thm:arboricity:logn}}
The theorem follows from Lemmas \ref{lem:E_alpha} and \ref{lem:E_alpha_alg} and taking $\alpha=\mu+1=6c$.  Observe that the space cost of Algorithm~\ref{alg:ealpha} can be bounded: 
we have $\log_{1+\eps} m$ levels where each level runs at most $\tau$ concurrent $\alpha$-good tests otherwise it will be terminated. Each $\alpha$-good test keeps an edge
and two counters and as result it occupies $O(1)$ space. Consequently     
the space usage of the algorithm is bounded by $O(\tau\log_{1+\eps} m)$. The space
bound in the theorem follows from the facts that $\tau=O(\frac{c}{\eps^2}\log n)$ for $\alpha=6c$ and $m \le cn$. 
\end{proofof}

\newcommand{\Proc}{Proceedings of the~}

\newcommand{\STOC}{Annual ACM Symposium on Theory of Computing (STOC)}
\newcommand{\FOCS}{IEEE Symposium on Foundations of Computer Science (FOCS)}
\newcommand{\SODA}{Annual ACM-SIAM Symposium on Discrete Algorithms (SODA)}
\newcommand{\SOCG}{Annual Symposium on Computational Geometry (SoCG)}
\newcommand{\ICALP}{Annual International Colloquium on Automata, Languages and Programming (ICALP)}
\newcommand{\ESA}{Annual European Symposium on Algorithms (ESA)}
\newcommand{\CCC}{Annual IEEE Conference on Computational Complexity (CCC)}
\newcommand{\RANDOM}{International Workshop on Randomization and Approximation Techniques in Computer Science (RANDOM)}
\newcommand{\APPROX}{International Workshop on  Approximation Algorithms for Combinatorial Optimization Problems  (APPROX)}
\newcommand{\PODS}{ACM SIGMOD Symposium on Principles of Database Systems (PODS)}
\newcommand{\SSDBM}{ International Conference on Scientific and Statistical Database Management (SSDBM)}
\newcommand{\ALENEX}{Workshop on Algorithm Engineering and Experiments (ALENEX)}
\newcommand{\BEATCS}{Bulletin of the European Association for Theoretical Computer Science (BEATCS)}
\newcommand{\CCCG}{Canadian Conference on Computational Geometry (CCCG)}
\newcommand{\CIAC}{Italian Conference on Algorithms and Complexity (CIAC)}
\newcommand{\COCOON}{Annual International Computing Combinatorics Conference (COCOON)}
\newcommand{\COLT}{Annual Conference on Learning Theory (COLT)}
\newcommand{\COMPGEOM}{Annual ACM Symposium on Computational Geometry}
\newcommand{\DCGEOM}{Discrete \& Computational Geometry}
\newcommand{\DISC}{International Symposium on Distributed Computing (DISC)}
\newcommand{\ECCC}{Electronic Colloquium on Computational Complexity (ECCC)}
\newcommand{\FSTTCS}{Foundations of Software Technology and Theoretical Computer Science (FSTTCS)}
\newcommand{\ICCCN}{IEEE International Conference on Computer Communications and Networks (ICCCN)}
\newcommand{\ICDCS}{International Conference on Distributed Computing Systems (ICDCS)}
\newcommand{\VLDB}{ International Conference on Very Large Data Bases (VLDB)}
\newcommand{\IJCGA}{International Journal of Computational Geometry and Applications}
\newcommand{\INFOCOM}{IEEE INFOCOM}
\newcommand{\IPCO}{International Integer Programming and Combinatorial Optimization Conference (IPCO)}
\newcommand{\ISAAC}{International Symposium on Algorithms and Computation (ISAAC)}
\newcommand{\ISTCS}{Israel Symposium on Theory of Computing and Systems (ISTCS)}
\newcommand{\JACM}{Journal of the ACM}
\newcommand{\LNCS}{Lecture Notes in Computer Science}
\newcommand{\RSA}{Random Structures and Algorithms}
\newcommand{\SPAA}{Annual ACM Symposium on Parallel Algorithms and Architectures (SPAA)}
\newcommand{\STACS}{Annual Symposium on Theoretical Aspects of Computer Science (STACS)}
\newcommand{\SWAT}{Scandinavian Workshop on Algorithm Theory (SWAT)}
\newcommand{\TALG}{ACM Transactions on Algorithms}
\newcommand{\UAI}{Conference on Uncertainty in Artificial Intelligence (UAI)}
\newcommand{\WADS}{Workshop on Algorithms and Data Structures (WADS)}
\newcommand{\SICOMP}{SIAM Journal on Computing}
\newcommand{\JCSS}{Journal of Computer and System Sciences}
\newcommand{\JASIS}{Journal of the American society for information science}
\newcommand{\PMS}{ Philosophical Magazine Series}
\newcommand{\ML}{Machine Learning}
\newcommand{\DCG}{Discrete and Computational Geometry}
\newcommand{\TODS}{ACM Transactions on Database Systems (TODS)}
\newcommand{\PHREV}{Physical Review E}
\newcommand{\NATS}{National Academy of Sciences}
\newcommand{\MPHy}{Reviews of Modern Physics}
\newcommand{\NRG}{Nature Reviews : Genetics}
\newcommand{\BullAMS}{Bulletin (New Series) of the American Mathematical Society}
\newcommand{\AMSM}{The American Mathematical Monthly}
\newcommand{\JAM}{SIAM Journal on Applied Mathematics}
\newcommand{\JDM}{SIAM Journal of  Discrete Math}
\newcommand{\JASM}{Journal of the American Statistical Association}
\newcommand{\AMS}{Annals of Mathematical Statistics}
\newcommand{\JALG}{Journal of Algorithms}
\newcommand{\TIT}{IEEE Transactions on Information Theory}
\newcommand{\CM}{Contemporary Mathematics}
\newcommand{\JC}{Journal of Complexity}
\newcommand{\TSE}{IEEE Transactions on Software Engineering}
\newcommand{\TNDE}{IEEE Transactions on Knowledge and Data Engineering}
\newcommand{\JIC}{Journal Information and Computation}
\newcommand{\ToC}{Theory of Computing}
\newcommand{\MST}{Mathematical Systems Theory}
\newcommand{\Com}{Combinatorica}
\newcommand{\NC}{Neural Computation}
\newcommand{\TAP}{The Annals of Probability}
\newcommand{\TCS}{Theoretical Computer Science}
\newcommand{\IPL}{Information Processing Letter}
\newcommand{\Algorithmica}{Algorithmica}

\bibliographystyle{plain}
\bibliography{matching,gem}{}

\begin{thebibliography}{10}

\bibitem{AGM12a}
K.~J. Ahn, S.~Guha, and A.~McGregor.
\newblock Analyzing graph structure via linear measurements.
\newblock In {\em \Proc 23rd \SODA}, pages 459--467, 2012.

\bibitem{BS15}
M.~Bury and C.~Schwiegelshohn.
\newblock Sublinear estimation of weighted matchings in dynamic data streams.
\newblock In {\em \Proc 23rd \ESA}, pages 263--274, 2015.

\bibitem{CCEHMMV16}
R.~Chitnis, G.~Cormode, H.~Esfandiari, M.T. Hajiaghayi, A.~McGregor,
  M.~Monemizadeh, and S.~Vorotnikova.
\newblock Kernelization via sampling with applications to finding matchings and
  related problems in dynamic graph streams.
\newblock In {\em \Proc 27th \SODA}, pages 1326--1344, 1326-1344.

\bibitem{CS14}
M.~Crouch and D.~S. Stubbs.
\newblock Improved streaming algorithms for weighted matching, via unweighted
  matching.
\newblock In {\em \Proc 17th \RANDOM}, pages 96--104, 2014.

\bibitem{DP14}
R.~Duan and S.~Pettie.
\newblock Linear-time approximation for maximum weight matchings.
\newblock {\em \JACM}, 61(1):1--23, 2014.

\bibitem{EKMS12}
S.~Eggert, L.~Kliemann, P.~Munstermann, and A.~Srivastav.
\newblock Bipartite graph matchings in the semi-streaming model.
\newblock {\em \Algorithmica}, 63(1-2):490--508, 2012.

\bibitem{EggertKMS12}
Sebastian Eggert, Lasse Kliemann, Peter Munstermann, and Anand Srivastav.
\newblock Bipartite matching in the semi-streaming model.
\newblock {\em Algorithmica}, 63(1-2):490--508, 2012.

\bibitem{EpsteinLMS11}
Leah Epstein, Asaf Levin, Juli{\'a}n Mestre, and Danny Segev.
\newblock Improved approximation guarantees for weighted matching in the
  semi-streaming model.
\newblock {\em SIAM J. Discrete Math.}, 25(3):1251--1265, 2011.

\bibitem{EHLMO15}
H.~Esfandiari, M.T. Hajiaghyi, V.~Liaghat, M.~Monemizadeh, and K.~Onak.
\newblock Streaming algorithms for estimating the matching size in planar
  graphs and beyond.
\newblock In {\em \Proc 26th \SODA}, 2015.

\bibitem{EsfandiariHLMO15}
Hossein Esfandiari, Mohammad~Taghi Hajiaghayi, Vahid Liaghat, Morteza
  Monemizadeh, and Krzysztof Onak.
\newblock Streaming algorithms for estimating the matching size in planar
  graphs and beyond.
\newblock In {\em Proceedings of the Twenty-Sixth Annual {ACM-SIAM} Symposium
  on Discrete Algorithms, {SODA} 2015, San Diego, CA, USA, January 4-6, 2015},
  pages 1217--1233, 2015.

\bibitem{FKMSZ05}
J.~Feigenbaum, S.~Kannan, A.~McGregor, S.~Suri, and J.~Zhang.
\newblock On graph problems in a semi-streaming model.
\newblock {\em \TCS}, 348(2):207--216, 2005.

\bibitem{GoelKK12}
Ashish Goel, Michael Kapralov, and Sanjeev Khanna.
\newblock On the communication and streaming complexity of maximum bipartite
  matching.
\newblock In {\em \Proc 23rd \SODA}, pages 468--485, 2012.

\bibitem{GoldreichR02}
Oded Goldreich and Dana Ron.
\newblock Property testing in bounded degree graphs.
\newblock {\em Algorithmica}, 32(2):302--343, 2002.

\bibitem{Han08}
Yijie Han.
\newblock Matching for graphs of bounded degree.
\newblock In {\em Frontiers in Algorithmics, Second Annual International
  Workshop, {FAW} 2008, Changsha, China, June 19-21, 2008, Proceeedings}, pages
  171--173, 2008.

\bibitem{han2008matching}
Yijie Han.
\newblock Matching for graphs of bounded degree.
\newblock In {\em Frontiers in Algorithmics}, pages 171--173. Springer, 2008.

\bibitem{HKNO09}
A.~Hassidim, J.~A. Kelner, H.~N. Nguyen, and K.~Onak.
\newblock Local graph partitions for approximation and testing.
\newblock In {\em \Proc 50th \FOCS}, pages 22--31, 2009.

\bibitem{HopcroftK73}
John~E. Hopcroft and Richard~M. Karp.
\newblock An n\({}^{\mbox{5/2}}\) algorithm for maximum matchings in bipartite
  graphs.
\newblock {\em {SIAM} J. Comput.}, 2(4):225--231, 1973.

\bibitem{Kap13}
M.~Kapralov.
\newblock Better bounds for matchings in the streaming model.
\newblock {\em \Proc 23rd \SODA}, pages 1679--1697, 2013.

\bibitem{KKS14}
M.~Kapralov, S.~Khanna, and M.~Sudan.
\newblock Approximating matching size from random streams.
\newblock {\em \Proc 25th \SODA}, pages 734--751, 2014.

\bibitem{KVV90}
R.~M. Karp, U.~V. Vazirani, and V.~V. Vazirani.
\newblock An optimal algorithm for on-line bipartite matching.
\newblock {\em \Proc 22nd \STOC}, pages 352--358, 1990.

\bibitem{KMM12}
C.~Konrad, F.~Magniez, and C.~Mathieu.
\newblock Maximum matching in semi-streaming with few passes.
\newblock In {\em \Proc 11th \RANDOM}, pages 231--242, 2012.

\bibitem{Kost84}
Alexandr~V. Kostochka.
\newblock Lower bound of the hadwiger number of graphs by their average degree.
\newblock {\em Combinatorica}, 4(4):307--316, 1984.

\bibitem{LP86}
L.~Lovasz and M.D. Plummer.
\newblock Matching theory.
\newblock In {\em North-Holland, Amsterdam-New York}, 1986.

\bibitem{McGregor:05}
A.~McGregor.
\newblock Finding graph matchings in data streams.
\newblock In {\em \Proc 8th \RANDOM}, pages 170--181, 2005.

\bibitem{McGregor:09}
A.~McGregor.
\newblock Graph mining on streams.
\newblock In {\em Encyclopedia of Database Systems}, pages 1271--1275.
  Springer, 2009.

\bibitem{MV16}
A.~McGregor and S.~Vorotnikova.
\newblock Planar matching in streams revisited.
\newblock In {\em \Proc 19th \APPROX}, 2016.

\bibitem{MV80}
S.~Micali and V.~V. Vazirani.
\newblock An $o(\sqrt{|V|} |e|)$ algorithm for finding maximum matching in
  general graphs.
\newblock {\em \Proc 21st \FOCS}, pages 17--27, 1980.

\bibitem{NW61}
C.~St. J.~A. Nash-Williams.
\newblock Edge-disjoint spanning trees of finite graphs.
\newblock {\em Journal of the London Mathematical Society}, 36(1):445--450,
  1961.

\bibitem{NW64}
C.~St. J.~A. Nash-Williams.
\newblock Decomposition of finite graphs into forests.
\newblock {\em Journal of the London Mathematical Society}, 39(1):12, 1964.

\bibitem{NO08}
H.~N. Nguyen and K.~Onak.
\newblock Constant-time approximation algorithms via local improvements.
\newblock In {\em \Proc 49th \FOCS}, pages 327--336, 2008.

\bibitem{P60}
List of open problems in sublinear algorithms: Problem 60.
\newblock \url{http://sublinear.info/60}.

\bibitem{Zelke12}
Mariano Zelke.
\newblock Weighted matching in the semi-streaming model.
\newblock {\em \Algorithmica}, 62(1-2):1--20, 2012.

\end{thebibliography}

\end{document}